\documentclass[11pt]{article}
\usepackage[protrusion=false]{microtype}
\usepackage{graphicx}
\usepackage{amsfonts}
\usepackage{amsmath}
\usepackage{amssymb}
\usepackage{amsthm}
\usepackage{geometry}
\usepackage{hyperref} \hypersetup{hidelinks}

  \newtheorem{theorem}{Theorem}[section]
  \newtheorem{corollary}[theorem]{Corollary}

  \newtheorem{definition}[theorem]{Definition}

\newcommand{\congruence}[3]{\ensuremath{{#1}\equiv {#2}\left(\bmod{#3}\right)}}

\newcommand{\dtime}{\mbox{\rm DTIME}}

\newcommand{\sharpp}{\mbox{\rm \#P}}
\newcommand{\sat}{\mbox{\rm SAT}}

\newcommand{\pe}{\mbox{\rm P}}
\newcommand{\p}{\mbox{\rm P}}

\newcommand{\np}{\mbox{\rm NP}}\newcommand{\NP}{\mbox{\rm NP}}
\newcommand{\optp}{\mbox{\rm OptP}}

\newcommand{\naturals}{\mathbb{N}}

\title{Team Diagonalization\thanks{This article will appear
as the Complexity Theory Column article in the 
September 2018 issue 
of~\emph{SIGACT News}~\protect\cite{hem-spa:jtoappear:team-diagonalization},
in
honor of the career of Professor Richard Ladner, who 
retired in 2017.}}
\author{Lane A. Hemaspaandra\thanks{This work was done in part
at ETH Z\"{u}rich's Department of
    Computer Science
    while on a sabbatical stay, which
was generously supported
    by that department.}\\
Dept.\ of Computer Science\\University of Rochester\\ Rochester, NY
  14627\\USA
\and
Holger Spakowski\thanks{This work was done in part while visiting the University of Rochester.}\\
Department of Mathematics\\~~~and Applied Mathematics\\
  University of Cape Town\\Rondebosch 7701, South Africa}

\date{July 28, 2018}

\begin{document}

\sloppy

\maketitle

\begin{abstract}
  Ten years ago, 
  Gla{\ss}er, Pavan, Selman, and Zhang~\cite{gla-pav-sel-zha:j:splitting} 
  proved that if $\pe \not= \np$, then all NP-complete
  sets can be simply split into two NP-complete
  sets.

That advance might naturally make one wonder about a quite different
potential consequence of NP-completeness:
Can the union of easy NP sets ever be hard?  In particular, can the
  union of  two
  non-NP-complete NP sets ever be NP-complete?  

Amazingly, Ladner~\cite{lad:j:np-incomplete} resolved 
this more than forty years ago:
If $\pe \not= \np$, then 
  all NP-complete sets can be
simply 
split into two non-NP-complete NP sets.
Indeed, this
  holds even when one requires the two non-NP-complete NP sets to be
  disjoint.

  We present this result as a mini-tutorial.  We give
   a relatively detailed proof
  of this result, using the same technique and idea
  Ladner~\cite{lad:j:np-incomplete} invented and used in proving a
  rich collection of results that include many that are more general
  than this result: delayed diagonalization.  In particular, the proof
  presented is based on what one can call \emph{team diagonalization} (or if 
one is being playful, perhaps even \emph{tag-team diagonalization}):
  Multiple sets are formed separately by delayed diagonalization, yet
  those diagonalizations are mutually aware and delay some of their
  actions until their partner(s) have also succeeded in some
  coordinated action.

  We relatedly note that, as a consequence of Ladner's result, 
if $\pe \not= \np$,
there exist OptP functions $f$ and $g$ whose 
\emph{composition} is NP-hard yet 
neither $f$ nor $g$ is NP-hard.
\end{abstract}

\section{Introduction}

This paper is about the fact 
that if $\pe \not= \np$, then NP-hardness can be
built by in a simple way by composing/combining non-NP-hard subparts.  

Our initial interest in this came from a question not about 
sets but about functions.
At a computational biology talk one of us attended, two actions were
sequentially taken on the input and the overall transformation was
clearly NP-hard.  During a discussion of the talk, the question came
up of whether that meant that at least one of the two composed
transformations must itself be NP-hard.  Although for the example of
that talk it probably was the case that one of the constituent
transformations could be directly proved NP-hard, 
what we eventually realized (not nearly as quickly as we should have)
is that (assuming $\pe \not= \np$) this is
not a general behavior.  In particular,
if $\pe \not=
\np$, then  there exist NP optimization functions (i.e., $\optp$ functions)
$f$ and $g$
such that 
neither $f$ nor $g$
is even NP-Turing-hard, yet $g(f(x))$ is (functional) NP-many-one-hard.
That is, non-NP-hard functions can via composition achieve NP-hardness.  

The natural way to show this is simply to observe that the function
case follows immediately from the language case, \emph{and the
  language case already put to bed in the 1970s by Ladner
  \cite{lad:j:np-incomplete}!}  So this paper presents, as a
mini-tutorial on delayed diagonalization, and in honor of Professor
Ladner's career on the occasion of his retirement, a proof 
of the language-case result, which is 
in fact a special case of one of his 1975 results.
(The paper  also provides the observation
that the (OptP) function case follows immediately from that.)

\subparagraph{Credit Where Credit Is Due and Blame Where Blame Is Due}
It is important to mention up front that, as this is a tutorial, the
credit for the results here belongs not to the authors of this
tutorial, but to the author of the underlying paper, namely Richard
Ladner.  All of the results of this paper are either explicitly in his
seminal paper on delayed diagonalization or are implicit from
or (for the case of 
Section~\ref{sect:functions}) follow easily from
its
results and techniques.  Similarly, though the particular proof write-up of
Theorem~\ref{thm:main} and the framing of
the theorem are our attempts to frame a clear yet very accessible 
proof for a tutorial article, it is very important to stress that 
all we are doing is employing Ladner's breakthrough technique, 
delayed diagonalization, that he developed to prove results of this
sort and of many related sorts, and indeed the proof is 
his in every sense other than that any errors our version might have,
which of course are ours.
In fact, in some sense we are 
giving a rather long proof to get a result 
weaker than ones that his original paper gets;
the reason for this is that our goal here is to be 
an expository paper, and to make as clear as possible, a particular
flavor of delayed diagonalization.  So in brief,
all the ideas here and much of the detail is 
due to the seminal paper of Ladner---except for 
any errors here, as those will
be due to flaws on our part in writing this tutorial.\footnote{Some Fine Print
and Disclaimers: Ladner's
article is in part focused on Turing reductions, as was common 
back then.  Yet it also weaves in explicitly, and implies,
a wide range of results about many-one reductions.  
In particular, the reader wanting to
see Ladner's 
far more general treatment of 
the type of thing we cover here is enthusiastically pointed to 
Ladner's seminal paper~\cite{lad:j:np-incomplete}, and 
in particular to, on his page~160, his Theorem~2 
which is a quite general version of the 
central result that we are presenting in this article, i.e., 
his 
Theorem~2 for the case of $A$ being any P set
and $B$ being an NP-complete (i.e., NP-many-one-complete)
set is, in effect, the special case being presented 
in Section~\ref{sect:sets} as 
Theorem~\ref{thm:main},
give or take the fact that 
in Theorem~\ref{thm:main} here we have put right into the 
theorem statement some details of the simple splitting's framework, 
namely, 
the polynomial-time function $r$ and how it controls the splitting.
Ladner's Corollary~2.1 on page~160 
of his paper gives the many-one analogue of 
his Theorem~2.}

\subparagraph{Organization} The rest of this article is organized as follows.  
Section~\ref{s:lit} mentions some related or contrasting 
results.
Section~\ref{s:defs} provides some preliminaries and definitions.
Section~\ref{sect:sets} contains our write-up, but using Ladner's 
delayed diagonalization, proving the result, due to Ladner,
that if P and NP differ, then there exists a pair of disjoint, non-NP-complete 
sets in $\np - \pe$ whose union is NP-complete.  
The
case of function composition described above will then follow easily,
and is covered in Section~\ref{sect:functions}, and 
Section~\ref{s:gen} provides a brief teaser for some of the 
rest of the world of results that Ladner's work provides and/or
underpins.

\section{Related and Contrasting 
Work}\label{s:lit}

The fascinating related work of Gla{\ss}er et
al.~\cite{gla-sel-tra-wag:j:union-disjoint} has a different focus.
That paper on the complexity of unions of disjoint sets primarily
focuses on whether unions of disjoint NP-complete sets remain hard.
That paper does have a section---Section 4.2 in that paper's
numbering---where the union of two disjoint sets is harder than its
components, but the results of that section do not overlap at all with
Ladner's~\cite{lad:j:np-incomplete} 
 work, due to Gla{\ss}er et al.'s focus on
 \emph{equivalent} (to each other) constituent sets.

Another related paper by Gla{\ss}er et al.~\cite{gla-pav-sel-zha:j:splitting} 
shows that every nontrivial 
(in the sense of the set and its complement each
containing at least two elements)
NP-complete set is (so-called)
m-mitotic.  This
result is interesting for us here because it implies that every nontrivial
NP-complete set can be partitioned into two P-separable sets that are
NP-complete.  However, while Ladner's work splits NP-complete sets
into 
NP-\emph{non}complete NP sets, 
Gla{\ss}er et al.\ are interested in
splitting NP-complete sets into NP-{\em complete} sets. 

We also mention the work of Hemaspaandra et
al.~\cite{hem-jia-rot-wat:j:join-lowers} that shows that the join
operator (the marked union operator) can yield a set of lower
complexity (in the extended low hierarchy) than either of its
constituents.  This regards a focus \emph{opposite that of 
Ladner's work presented in this paper}.
In particular, 
that work is about using 
combinations to lower complexity; the 
focus of the Ladner work that we are presenting 
is on using combinations to rise from
non-NP-completeness to NP-completeness.

We prove Ladner's key result using {\em team 
  diagonalization}.  Delayed diagonalization is the powerful technique
first used by Ladner~\cite{lad:j:np-incomplete} (see also,
e.g.,~\cite{koz:j:subrecursive,sch:j:uni,reg:j:diag-unif-fixed-point,for:j:diagonalization,reg-vol:j:gap}) 
to show that if $\pe \not= \np$, then there are incomplete
sets in $\np - \pe$.  
Many
people, rather naturally given which example one 
tends to see in courses and textbooks, think of delayed diagonalization 
as 
the technique of having one set handle
a list (actually two lists) of requirements by looking deeply into its
own history.  But in fact, delayed diagonalization is far more 
flexible and powerful than merely being able to do that.
In (what we here will call) 
team diagonalization, though, Ladner in effect has two
sets, each with its own list of requirements to
satisfy, but the two sets will take long turns as to which of
them is working on its requirements, and while one is doing that,
the other will politely remain simple and boring.  Loosely put, the
sets will each respect the goals of the other set, and
will take on burdens in a completely coordinated ``lock-step''
fashion. 

\section{Preliminaries}\label{s:defs}
For each string $x$, the number of characters in $x$ will be 
denoted $|x|$.
For each set $A$ and each natural number $k$, $A^{=k}$ will
denote $\{ x ~|~ x \in A \land |x| = k\}$.
We take all sets and classes to be with respect to 
the 
alphabet $\Sigma = \{0,1\}$.
The symmetric difference 
operation for sets,  
$(A-B)  \cup  (B-A)$,
will be denoted 
$A\bigtriangleup B$.

All logarithms in this paper are base two, e.g., $\log\log i$ means
$\log_2(\log_2(i))$.

We say $A \le^p_m B$ exactly if $A$ is polynomial-time many-one
   reducible to $B$ (i.e., there is a polynomial-time function
$f$ such that, for each $x$, it holds
that $x\in A \iff f(x)\in B$).
A set $B$ is NP-hard (with respect to $\le^p_m$
  reductions) exactly if for all $A \in \np$, $A \le^p_m B$.
A set $B$ is NP-complete (with respect to $\le^p_m$
  reductions) exactly if $B\in\np$ and $B$ is NP-hard.
We say $A \le^p_T B$ if $A \in \pe^B$.
A set $B$ is NP-hard with respect to $\le^p_T$
  reductions exactly if for all $A \in \np$, $A \le^p_T B$.
A set $B$ is NP-complete with respect to $\le^p_T$
  reductions exactly if $B\in\np$ and $B$ is NP-hard with 
respect to $\le^p_T$ reductions.

Sets $A$ and $B$ are said to be P-separable exactly if
there exists a set $L \in \pe$ such that $A \subseteq L \subseteq 
\overline{B}$, i.e., there is a polynomial-time set that is a 
(possibly nonstrict) superset 
of $A$ yet has no intersection with $B$.

 Given functions $f$ and $g$, we say that $f \le ^p_m g$ ($f$
 functional many-one reduces to $g$) if there are polynomial-time functions
 $h_1$ and $h_2$ such that, for all $x$, 
it holds that 
$f(x) = h_2(g(h_1(x)))$~\cite{zan:j:sharp-p}.
Functional many-one reductions are even more restrictive  than
metric reductions~\cite{kre:j:optimization} 
(which have almost the same definition except they 
allow $h_2$ to have direct access to $x$, i.e., the definition's key
part is 
$f(x) = h_2(x,g(h_1(x)))$)
and are most commonly studied in the context of
$\sharpp$-completeness in order to prove stronger completeness results
than mere $\sharpp$-metric-completeness or $\sharpp$-Turing
completeness.
For example, Valiant's
notion of $\sharpp$-completeness 
in his seminal 
paper~\cite{val:j:permanent} 
is the Turing-reduction notion, and for the permanent of $(0,1)$ matrices 
the reduction he builds is not a many-one reduction
(although for the permanent of $(-1,0,1,2,3)$ matrices
Valiant
does establish 
$\sharpp$-many-one-completeness);
$\sharpp$-many-one-completeness for the permanent of $(0,1)$-matrices
was obtained only more than a decade later, 
by 
Zank{\'{o}}~\cite{zan:j:sharp-p}.
As another example,
Deng and Papadimitriou's~\cite{den-pap:j:solution-concepts}
proof that the Shapley--Shubik power index 
is $\sharpp$-metric-complete was later
strengthened to a $\sharpp$-many-one-completeness 
result~\cite{fal-hem:j:power-compare}.

Let $\chi_A$ denote the characteristic function of $A$, that is,
$\chi_A(x)$ equals 0 if $x\not\in A$ and equals 1 if $x\in A$.
We say a function $g$ is NP-hard exactly if, for every
$A\in\np$, $\chi_A \le^p_m g$.  (Note that this is equivalent to
$\chi_{SAT} \le^p_m g$.) We say that a function $g$ is NP-Turing-hard
exactly if, for all $A\in\np$, $A \in \pe^g$ (equivalently, $\sat \in \pe^g$).

We remark that in the literature when ``NP-hard functions'' are
mentioned, the term often means what we here call NP-Turing-hard.
However, for clarity, in this paper we always use, for both sets and
functions, the terms NP-hard (for the many-one case) and
NP-Turing-hard (for the Turing case).  
The key theorems statements here---in particular,
Theorems~\ref{thm:main},
\ref{t:optp},
and
\ref{thm:exp}---are 
stated for the most demanding choices regarding many-one versus
Turing reductions (even when this requires mixing and 
matching within the theorem statements, as indeed happens in 
each of these three theorem statements), 
and so they imply the weaker choices. 

\section{Splitting NP-Complete Sets into NP-Incomplete Sets}  \label{sect:sets}

The following captures Ladner's 
beautiful insight that
if $\pe \not= \np$, then each NP-complete set
can be built from (or looked at from the 
other direction, partitioned into) 
two non-NP-complete NP sets, in a very simple way.

\begin{theorem}%
\label{thm:main}
If 
$\pe \not= \np$ and 
$S$ is an 
$\np$-complete set,
then there is a function $r: \naturals \rightarrow \naturals$ 
such that, for all $n$, $r(n)$ can be computed in time polynomial in
$n$ and the following disjoint sets $A$ and $B$ 
belong to $\np - \pe$, satisfy
$A \cup B = S$,
 and 
are not $\np$-hard even under
polynomial-time Turing reductions:
\begin{equation} \label{eq:def-of-A}
  A = \left\{ x ~|~ x \in S \mbox{ and } r(|x|) \mbox{ is even} \right\}
\end{equation}
and 
\begin{equation}  \label{eq:def-of-B}
  B = \left\{ x ~|~ x \in S \mbox{ and } r(|x|) \mbox{ is odd} \right\}.
\end{equation}
\end{theorem}
\begin{proof}
Suppose $\pe \not= \np$.  Let 
$M_1, M_2, \ldots $ be a standard enumeration of deterministic
oracle Turing machines, each of which is explicitly polynomially
clocked (upper-bounded) independently of its oracle. 
We assume, w.l.o.g., that the enumeration and the clocking are such
that there is a universal oracle Turing machine
$\mathcal{U}$ such that the following is true:
\begin{enumerate}
\item For each $j \geq 1$, each $x$, and each
$A$,   $\mathcal{U}^A(x,j)$ simulates
$M^A_j(x)$ (in the sense that $x \in L(M_j^A) \iff 
x \in L(\mathcal{U}^A(x,j))$), and 
\item For each $j \geq 1$,
  each $x$, and each $A$, $\mathcal{U}^A(x,j)$ runs in time at 
most $|x|^j + j$.\footnote{Technical aside (that we suggest ignoring 
during a first reading): The w.l.o.g.s~in that 
sentence and a few lines above are the kind of somewhat painful
groundwork that is often skipped over---as we also will mostly do
here.  But, to touch for those interested on what is 
under the hood:  The natural way to build a nice clocked enumeration is
to pair each machine (from a truly standard enumeration of all Turing
machines) together with every possible clock drawn from some nice family of
polynomials that for every polynomial has at least one member that
majorizes that polynomial over all natural numbers.  The typical
family used for that is $n^k + k$.  One needs to then interlace and
assign numbers in the enumeration so that even the earliest members of
the enumeration meet their own claimed time bounds (such as 
that the $k$th machine will run in time $n^k + k$), and doing that
often involves delaying (not at all in the same sense
of the word as in delayed diagonalization) 
when machines come in (i.e., making sure their location in 
the enumeration has a high enough number) or even making some
machines in our list dummy machines that ignore everything and in one
step halt.  And in doing what was just described, one has to account
for the fact that clocking the time of a machine itself has overhead,
since one puts a supervisor on top of the machine; but the overhead is
mild---certainly at most polynomial.  Beyond that, here we want to be
able to uniformly simulate any given machine on the fly, and so one
also has to take into account the cost of the universal machine's own
simulation of machines---itself also incurring a mild overhead---so
that whatever claim one wants to make 
about the universal machine's running times is correct.
Despite that, all the w.l.o.g.~claims above indeed can be routinely
achieved.  Important in that is that each (deeply) underlying machine is
paired with polynomials (infinitely many)
greater than particular given polynomial $p$, and so
basically each machine appears infinitely often on the listing, and indeed
occurs with any particular needed polynomial ``headroom for
simulation'' even on top of the underlying time cost.  So if an
(deeply) underlying machine with a given oracle $A$ does happen to run in some
polynomial time, one of the machines in our enumeration will capture
that in such a way that even its simulation within the universal
machine will duplicate, under oracle $A$, the action of the underlying
machine in terms of acceptance and rejection, and will do so 
relative to our time claims without 
being cut off by time issues.}  
\end{enumerate}

Let $S$ be an an arbitrary NP-complete set.  Let $c\ge 1$ be a constant such that $S
\in \dtime(2^{n^c})$.
Let $A$ and $B$ be defined by Eqns.~(\ref{eq:def-of-A}) and (\ref{eq:def-of-B});
of course, for those definitions to be meaningful, we will
need to define $r$.
In particular, we 
will now define $r$ such that $r$ is nondecreasing, and
$r(n)$ can be computed in time polynomial in $n$.  
It follows, keeping in mind
that $S\in \np$, that $A \in \np$ and 
$B \in \np$.  It also follows, from the definition of 
$A$ and $B$, that $A$ and $B$ are disjoint and satisfy
$A \cup B = S$.
In addition, our definition of $r$ will be such that 
the other claims in the conclusion of Theorem~\ref{thm:main} 
hold, namely, that neither $A$ nor $B$ belongs to $\p$, 
and neither $A$ nor $B$ is NP-hard even under 
polynomial-time Turing reductions.

Let $r(0) = r(1) = r(2) = 2$.
To define $r$, we describe a procedure that, for any $i \ge 2$, computes the value 
$r(i+1)$ in time polynomial in $i$ based on the values the
values $r(0), r(1), \ldots , r(i)$.

\medskip

\noindent {\bf Computation of  $\boldsymbol{r(i+1)}$ based
  on $\boldsymbol{r(0), r(1), \ldots , r(i) \; (i \ge 2):}$}

When we determine $r(i+1)$, we try to diagonalize against machine
$M_{\lfloor   r(i)/2\rfloor}$. If $r(i)$ is odd then we try to make sure that 
$M_{\lfloor   r(i)/2\rfloor}$ does not decide SAT with the help of
oracle $A$ and if $r(i)$ is even then we try to make sure that 
$M_{\lfloor   r(i)/2\rfloor}$ does not decide SAT with the help of
oracle $B$. 
If we succeed then we set $r(i+1) = r(i)+1$.
Otherwise, we set $r(i+1) = r(i)$.  (Then 
the same machine and oracle will be tried to diagonalize against  
when $r(i+2)$ is determined.) 

If 
\begin{equation} \label{eq:1}
  2^{{\left( (\log\log i)^{\lfloor r(i)/2 \rfloor } +  \lfloor r(i)/2 \rfloor  \right)}^c}  \ge i
\end{equation}
then the diagonalization fails, and we set $r(i+1) = r(i)$. 
Otherwise there are two cases, as follows.
(As one reads the cases, the fact that $A$ and $B$ are being used
 in text that in part is also creating them may seem circular.  But why
 this is not 
fatally  naughty is explained in the discussion justifying
 the correctness of the construction.)
\begin{description}
\item[Case 1: $\boldsymbol{r(i)} $ is odd.]

We try to diagonalize against $M_{\lfloor r(i)/2\rfloor}$ with oracle $A$.
Determine if there exists
a string $y$ of length at most $\log\log i$ satisfying
\begin{equation} \label{eq:diag-against-A}
  y \in \sat \Longleftrightarrow y \notin L(M_{\lfloor r(i)/2\rfloor}^A).
\end{equation}
If such a $y$ exists then
$\sat$ is not polynomial-time Turing reducible to 
$A$ via oracle machine $M_{\lfloor r(i)/2\rfloor}$.
Hence the diagonalization is successful, and so we set $r(i+1) = r(i)+1$.
Otherwise, we set $r(i+1) = r(i)$.  
\item[Case 2: $\boldsymbol{r(i)}$ is even.]

We try to diagonalize against $M_{\lfloor r(i)/2 \rfloor}$ with oracle $B$.
Determine if there exists
a string $y$ of length at most $\log\log i$ satisfying
\begin{equation} \label{eq:diag-against-B}
  y \in \sat \Longleftrightarrow y \notin L(M_{\lfloor r(i)/2\rfloor}^B).
\end{equation}
If such a $y$ exists then
$\sat$ is not polynomial-time Turing reducible to 
$B$ via oracle machine $M_{\lfloor r(i)/2\rfloor }$.
Hence the diagonalization is successful, and so 
we set $r(i+1) = r(i)+1$.
Otherwise, we set $r(i+1) = r(i)$. 
\end{description}
For any given $n$, we compute $r(n)$ as follows:
For $i = 2, 3, \ldots , n-1$, we successively compute $r(i+1)$
as described above
based on the values $r(0), r(1), \ldots , r(i-1), r(i)$.
Note that if for each $i \in \{ 2, 3, \ldots , n-1\}$ this is possible in
time polynomial in $i$, then $r(n)$ can be computed in time
polynomial in~$n$.

\bigskip

It remains to show that this construction is correct.

\begin{enumerate}
\item The construction of $r(i+1)$ in cases 1 and 2 above obviously depends on
  the oracle sets $A$ and $B$, which according to Eqn.~(\ref{eq:def-of-A}) and
  Eqn.~(\ref{eq:def-of-B}) depend on $r$ being odd or even.  However,
  the construction is not circular: To determine whether 
  $y \notin L(M_{\lfloor r(i)/2\rfloor}^{(\cdot)})$, we only need to
  know the oracle up to length
  $(\log\log i)^{\lfloor r(i)/2 \rfloor } + \lfloor r(i)/2 \rfloor$, which is smaller than $i$ by Eqn.~(\ref{eq:1}).
Hence the only $r(k)$ values that $r(i+1)$ may depend on 
are values $r(k)$ for $k$ less than than $i$.
\item\label{correct-time}
 We 
argue---and recall that, as noted above, showing this 
establishes
that 
$r(n)$ can be computed (from scratch) in time polynomial in~$n$---that 
for each $i\ge 2$, 
the procedure
  that determines $r(i+1)$ based on $r(0), r(1), \ldots , r(i-1), r(i)$
  runs in time polynomial in $i$.
  To this end, we have to show that 
  the conditions in
  Eqns.~(\ref{eq:diag-against-A}) and~(\ref{eq:diag-against-B}) can be
  checked for all $y$ with $|y| \le \log\log i$ in time polynomial in
  $i$.

  First, for each $y$, checking whether $y \in \sat$ can be done by brute
  force in time polynomial in
  $i$ since $y$ is of length at most $\log \log i$.

  Second, the running time of $M_{\lfloor r(i)/2\rfloor}^{(\cdot)}$ on inputs
  of length $\log \log i$ is at most  
  $t = (\log\log i)^{\lfloor r(i)/2 \rfloor } + \lfloor r(i)/2 \rfloor$.
  Hence the length of each oracle query $q$ is at most
  $t$.
  Since $S \in \dtime(2^{n^c})$, there exists a constant 
$s > 0$ such that 
we can determine 
in time
  \[
  s\cdot 2^{t^c} \leq s \cdot 2^{{\left( (\log\log i)^{\lfloor r(i)/2 \rfloor } + \lfloor r(i)/2 \rfloor  \right)}^c}
  \]
whether $q \in S$.
And since we got past the test of Eqn.~(\ref{eq:1}),
we thus know that the time used is at most $s \cdot i$.
  Finally, note that no more than $(2\log i) -1 $ 
different strings
  $y$ have to be checked in whichever one of 
  Eqns.~(\ref{eq:diag-against-A}) or~(\ref{eq:diag-against-B}) applies. 
  Hence
  the whole procedure of determining $r(i+1)$ based on 
  $r(0), r(1), \ldots , r(i)$
  takes time polynomial in $i$. 
\item
  In the above construction, we try to diagonalize against machine
$M_{\lfloor   r(i)/2\rfloor}$ when we determine $r(i+1)$.
  Hence to check that we eventually
  diagonalize
  against all deterministic polynomial-time oracle Turing machines
  $M_j^{(\cdot)}$, we only have to show that $r(i)$---which, recall,
is nondecreasing---grows indefinitely.

  Suppose that there exist $k$ and $n_0$ such that $r(n) = k$ for all $n \ge
  n_0$.  Then there are two cases:
  \begin{description}
   \item[Case 1: $\boldsymbol{k}$ is odd.]\ \\
     Then for all $i \ge n_0$ it holds that, for all $y$ with $|y| \le \log\log i$:
     \begin{align*}
         y \in \sat & \Longleftrightarrow y \in L(M_{\lfloor
           r(i)/2\rfloor}^A), \mbox{ and}
     \end{align*}
     \begin{align*}
          y \in L(M_{\lfloor
           r(i)/2\rfloor}^A) &
                    \Longleftrightarrow y \in L(M_{\lfloor
           k/2\rfloor}^A).
\end{align*}
So for every string $y$, it holds that 
$         y \in \sat \Longleftrightarrow 
y \in L(M_{\lfloor
           k/2\rfloor}^A)$.
     Furthermore, by construction, $A$ contains in this case only finitely many strings.
     It follows that SAT can be decided in polynomial time, which
     contradicts our assumption that $\pe \not= \np$.
  \item[Case 2: $\boldsymbol{k}$ is even.]\ \\
     Then for all $i \ge n_0$ it holds that, 
for all $y$ with $|y| \le \log\log i$:
     \begin{align*}
         y \in \sat &\Longleftrightarrow y \in L(M_{\lfloor
           r(i)/2\rfloor}^B) \mbox{ and}
     \end{align*}
     \begin{align*}
 y \in L(M_{\lfloor
           r(i)/2\rfloor}^B)
                    &\Longleftrightarrow y \in L(M_{\lfloor
           k/2\rfloor}^B).
     \end{align*}
So  for every string $y$, it holds that 
$         y \in \sat \iff 
y \in L(M_{\lfloor
           k/2\rfloor}^B)$.
     Furthermore, by construction, $B$ contains in this case only finitely many strings.
     It follows that SAT can be decided in polynomial time, which 
     contradicts our assumption that $\pe \not= \np$.

\end{description}

\item  Finally, we must argue that $A \not\in\pe$ and $B\not\in\pe$.   
(Recall that, at the start of the proof, we 
pointed out that both $A$ and $B$ belong to NP, that $A$ and $B$ are 
disjoint, and that $A\cup B = S$.  
So we do not need to re-argue those points here.)
Suppose for example that 
$A \in \pe$.   That implies, since $S = A \cup B$, that 
$\pe^S \subseteq \pe^B$.  So the fact that $S$ is NP-complete 
certainly yields that $B$ is NP-Turing-hard.  But that 
contradicts the fact that, as argued above, $B$ is not 
NP-Turing-hard.  So $A\not\in \pe$.  

By the analogous 
argument, it also holds that $B \not\in \pe$.~\qedhere
\end{enumerate}
\end{proof}

One has the following easy corollaries.
(For disjoint sets $A_1$ and $A_2$, note that 
$A_1\cup A_2 = A_1 \bigtriangleup A_2$, so one could equally well
in the theorems below make the theorems be stated not about a union
($A_1\cup A_2$)
but about 
a symmetric difference
($A_1 \bigtriangleup A_2$).)
\begin{corollary} \label{coroll:1}
$\pe \not= \np$ if and only if for each $\np$-complete set $S$ there
exist disjoint $\np$ sets $A_1 \subseteq S$ and $A_2 \subseteq S$ such
that $S = 
A_1 \cup A_2
$ and neither $A_1$ nor $A_2$ is $\np$-complete.
\end{corollary}
\begin{proof}
The direction from left to right follows 
directly from
Theorem~\ref{thm:main}.

Now suppose that $\pe = \np$.  Let $S$ be any set such that $S \not=
\Sigma^*$ and $S \not=\emptyset$.  Then $S$ is NP-complete.  The only
subset of $S$ that is not NP-complete is the empty set.  Hence there
are no NP-incomplete sets $A_1 \subseteq S$ and $A_2 \subseteq S$ such that $A_1
\cup A_2 = S$.~\end{proof}
If $\pe \not= \np$ then the sets $A_1$ and $A_2$ in
Corollary~\ref{coroll:1} are both not in P\@.

\begin{corollary} \label{coroll:2}
If $\pe \not= \np$, then every $\NP$-complete set $S$ has the property that
there is a $\pe$ set $D$ such that both $S \cap D$ and $S \cap
\overline{D}$ are $\NP$-incomplete (i.e., are in $\NP$ yet are not $\NP$-complete).

That is, for every $\NP$-complete set $S$ there exist $\pe$-separable sets
$A_1$ and $A_2$ such that $S = A_1 \cup A_2
$ and neither $A_1$ nor
$A_2$ is $\NP$-hard.
\end{corollary}

\begin{proof}
Let
\[
  D = \left \{ x \in 
\Sigma^*  ~|~ r(|x|) \mbox{ is even} \right\}.
\]
Note that
$A = S \cap D$ and $B = S \cap \overline{D}$ are the sets $A$ and $B$
in Theorem~\ref{thm:main}.~\end{proof}

We 
mention a result by 
Gla{\ss}er et
al.~\cite{gla-pav-sel-zha:j:splitting} that is similar in spirit to
Corollary~\ref{coroll:2} even though it is trying to achieve the opposite.
Gla{\ss}er et
al.~\cite{gla-pav-sel-zha:j:splitting} call a set 
$A$ nontrivial exactly if 
both $A$ and $\overline{A}$ contain at least two elements.
Gla{\ss}er et al. showed (among other things) that every
nontrivial set $A$ that is NP-complete is also m-mitotic. This easily
implies the following theorem.

\begin{theorem}[\cite{gla-pav-sel-zha:j:splitting}] \label{thm:glasser-splitting}
If $\pe \not= \np$, then every $\NP$-complete set $S$ has the property that
there is a $\pe$ set $D$ such that both $S \cap D$ and $S \cap
\overline{D}$ are $\NP$-complete.
\end{theorem}

While Corollary~\ref{coroll:2} shows (if $\pe \not= \np$)
that every NP-complete set $S$ can be split into P-separable sets that
are {\em not} $\np$-complete, Gla{\ss}er et al.'s theorem implies
that 
(regardless of whether or not 
$\pe \not= \np$)
every nontrivial 
NP-complete set $S$ can be split into P-separable sets such that
both sets {\em are} indeed NP-complete.

Note that regarding $\p = \np$ we have the following.

\begin{theorem} \label{thm-was-coroll:3}%
If $\pe = \np$, then no $\NP$-complete set $S$ has the property that
there is a $\pe$ set $D$ such that both $S \cap D$ and $S \cap
\overline{D}$ are $\NP$-incomplete (i.e., are in 
$\NP$ yet are not $\NP$-complete).
\end{theorem}

\begin{proof}
Suppose $\pe = \np$.  Then the only NP-incomplete sets are $\emptyset$
and 
$\Sigma^*$.
But the only sets that can be formed by the union of two sets 
chosen from $\{\emptyset,\Sigma^*\}$ are $\emptyset$ and $\Sigma^*$,
which as just mentioned are NP-incomplete, yet the theorem's 
claim is that $S$ is NP-complete.~\end{proof}
Note that Corollary~\ref{coroll:2} and Theorem~\ref{thm-was-coroll:3} are not
converses of each other.  They are actually stronger than just giving
an ``if and only'' statement.

The following notes that the sets being used are not merely 
disjoint in the sense that no string participates in both of the sets,
but also differ so strongly that no length participates in both of 
the sets.

\begin{definition}
  Two sets $A_1$ and $A_2$ are \emph{strongly disjoint} exactly if, for every
  $k$,  $A_1^{=k} = \emptyset$ or $A_2^{=k} = \emptyset$.
\end{definition}

\begin{corollary} \label{coroll:4}
$\pe \not= \np$ if and only if for each $\NP$-complete set $S$ there
exist strongly disjoint, $\pe$-separable, $\NP$ sets $A_1 \subseteq S$ and $A_2 \subseteq S$ such
that $A_1 \cup A_2 = S$ and neither $A_1$ nor $A_2$ is $\NP$-complete.
\end{corollary}
\begin{proof}
The direction from left to right follows easily from
Theorem~\ref{thm:main} because it is easy to see that the sets $A$ and
$B$ in that theorem are strongly disjoint P-separable NP-sets.

The direction from right to left follows directly from Corollary~\ref{coroll:1}.~\end{proof}

\section{Composition of Functions: Hard Functions Composed from
  Nonhard Ones} \label{sect:functions}

Consider NPTMs where each nondeterministic path outputs one value
(say, a string over the alphabet $
\Sigma^*
$).  Paths that do not
explicitly output a value are by convention viewed as outputting
$\epsilon$, the empty string.  A function $f$ is said to be in 
$\optp$~\cite{kre:j:optimization}
exactly if  there is a thus-viewed NPTM $N$ such that, for
each $x$, $f(x)$ is the lexicographically maximum value among all values 
output by paths of $N$ on input $x$.
(By lexicographical order, we mean as is standard the order in which
$\epsilon < 0 < 1 < 00 < 01 < 10 < 11 < 111 < \dots$.
$\optp$ is often viewed as having codomain $\naturals$ rather than
$\Sigma^*$, 
but by the natural bijection, the views are the same.)

We now observe, as an easy consequence of the main result of
the previous section, that the composition of non-NP-hard 
NP optimization functions can achieve
NP-hardness.  Let $g \circ f$ denote the composition of
the functions, i.e., the function defined by, for 
each $x$, $(g \circ f)(x) = g(f(x))$.

\begin{theorem}\label{t:optp}
If $\pe \not= \np$, then there exist $\optp$ functions $f$ and $g$
such that 
neither $f$ nor $g$ is $\NP$-hard (or
even $\NP$-Turing-hard), yet $g \circ f$ is an $\NP$-hard $\optp$ function.
\end{theorem}
\begin{proof}
  Let $S = \sat$ and let $A$ and $B$ be disjoint NP sets such that neither $A$ nor $B$ is
  NP-Turing-complete and $A \cup B=S$.  Such sets exist according to
  Theorem~\ref{thm:main}. 

  We define $f$ and $g$ as follows: \medskip

  \noindent For each $x \in 
\Sigma^*,$ \medskip 

  $
     f(x) = 
      \begin{cases}
         1^{|x|+1}  & \mbox{if } x \in A\\
         0x & \mbox{otherwise,} 
      \end{cases}$ \medskip

  \noindent  and for each $z \in 
\Sigma^*$,\medskip

  $  g(z) = 
      \begin{cases}
         1  & \mbox{if } (z\mbox{'s first bit is a 1})\mbox{ or } (z = 0x \mbox{ and } x \in B) \\
         0 & \mbox{otherwise.} 
      \end{cases}
  $ \medskip

  \noindent It is easy to see that for each $x \in 
\Sigma^*$,\medskip

  $
     g(f(x)) = 
      \begin{cases}
         1  & \mbox{if } x \in A\cup B \\
         0 & \mbox{otherwise.} 
      \end{cases}
  $\medskip

  Hence $\chi_{SAT} \le^p_m g\circ f$, and therefore $g\circ f$ 
is NP-hard under
  polynomial-time many-one reductions.  Furthermore, $f$, $g$, and $g\circ f$
  are easily seen to be $\optp$ functions.

  It remains to show that $f$ and $g$ are not NP-Turing-hard.  

  Suppose $f$ is
  NP-Turing-hard.  Then there exists a polynomial-time oracle machine $M$
  such that $\sat = L(M^f)$. Since given $A$ as an oracle it is easy to compute
  $f$, it follows that there is a polynomial-time oracle Turing
  machine $M'$ such that $\sat = L(M'^A)$.  But this is a
  contradiction because we assumed that $A$ is not NP-Turing-hard.

  In the same way, we can show that $g$ is not NP-Turing-hard.~\end{proof}

\section{Generalizations and Variants}\label{s:gen}

Theorem~\ref{thm:main} can be generalized/varied in many ways.
Please see Ladner's original paper~\cite{lad:j:np-incomplete} for 
results 
in such 
a general form that this can be done very 
broadly.  But for 
now let us mention a few examples. 
First, the theorem holds for many classes other than $\np$ 
(because Ladner proves it in a very general setting, see 
especially p.~160 
of his paper), for
instance for PSPACE, EXP, EXPSPACE, EEXP, $\Sigma_i^p$, $\Pi_i^p$, etc.
Second, Theorem~\ref{thm:main} says that we can split $S$ into two
sets that are NP-incomplete.  Clearly, a 
straightforward adaptation of the proof
shows that for any integer $k>1$, $S$ can be split into $k$
incomplete sets.  As an example, we mention a variation of 
Theorem~\ref{thm:main} for the case of EEXP, with things split into three
incomplete sets.  Of course, here one does 
not need any assumption of the form $\pe
\not= \np$.
\begin{theorem} \label{thm:exp}
Let $S$ be any 
{\rm{}EEXP}-complete set.
Then there is a function $r: \naturals \rightarrow \naturals$ 
such that for all $n$, $r(n)$ can be computed in time polynomial in
$n$ and the following pairwise-disjoint 
$\rm EEXP$ sets $A$, $B$, and $C$
satisfy $S = A\cup B \cup C$
and are not {\rm{}EEXP}-hard under
polynomial-time Turing reductions:
\[
  A = \left\{ x ~|~ x \in S \land (  \congruence{r(|x|)}{0}{3} ) \right\},
\]
\[
  B = \left\{ x ~|~ x \in S \land (  \congruence{r(|x|)}{1}{3} ) \right\}, 
\]
and
\[ 
  C = \left\{ x ~|~ x \in S \land (  \congruence{r(|x|)}{2}{3} ) \right\}.
\]
\end{theorem}
The proof differs from the proof of Theorem~\ref{thm:main} in that
one now has to consider three different cases: $r(|x|) \equiv 0
\pmod{3}$, $r(|x|) \equiv 1 \pmod{3}$, and $r(|x|) \equiv 2 \pmod{3}$,
and further, the length of the string $y$ that one uses in one's 
diagonalization
must be of length 
roughly triple-logarithmic in $i$, that is, computations have to
look even more deeply back within the history.

\subparagraph*{Acknowledgments.}  We thank Max Alekseyev and Daniel
\v{S}tefankovi\v{c} for helpful conversations.  We are 
grateful to Richard 
Ladner for his work that this paper is presenting, and for his 
lifetime's rich range of central contributions
to computer science, 
which have 
inspired us and so many others.

\bibliographystyle{alpha}

\begin{thebibliography}{GSTW08}

\bibitem[DP94]{den-pap:j:solution-concepts}
X.~Deng and C.~Papadimitriou.
\newblock On the complexity of comparative solution concepts.
\newblock {\em Mathematics of Operations Research}, 19(2):257--266, 1994.

\bibitem[FH09]{fal-hem:j:power-compare}
P.~Faliszewski and L.~Hemaspaandra.
\newblock The complexity of power-index comparison.
\newblock {\em Theoretical Computer Science}, 410(1):101--107, 2009.

\bibitem[For00]{for:j:diagonalization}
L.~Fortnow.
\newblock Diagonalization.
\newblock {\em Bulletin of the EATCS}, 71:102--112, 2000.

\bibitem[GPSZ08]{gla-pav-sel-zha:j:splitting}
C.~Gla{\ss}er, A.~Pavan, A.~Selman, and L.~Zhang.
\newblock Splitting {NP}-complete sets.
\newblock {\em SIAM Journal on Computing}, 37:1517--1535, 2008.

\bibitem[GSTW08]{gla-sel-tra-wag:j:union-disjoint}
C.~Gla{\ss}er, A.~Selman, S.~Travers, and K.~Wagner.
\newblock The complexity of unions of disjoint sets.
\newblock {\em Journal of Computer and System Sciences}, 74:1173--1187, 2008.

\bibitem[HJRW98]{hem-jia-rot-wat:j:join-lowers}
L.~Hemaspaandra, Z.~Jiang, J.~Rothe, and O.~Watanabe.
\newblock Boolean operations, joins, and the extended low hierarchy.
\newblock {\em Theoretical Computer Science}, 205(1--2):317--327, 1998.

\bibitem[HS18]{hem-spa:jtoappear:team-diagonalization}
L.~Hemaspaandra and H.~Spakowski.
\newblock Team diagonalization.
\newblock {\em SIGACT News}, 49(3), 2018.
\newblock To appear.

\bibitem[Koz80]{koz:j:subrecursive}
D.~Kozen.
\newblock Indexings of subrecursive classes.
\newblock {\em Theoretical Computer Science}, 11(3):277--301, 1980.

\bibitem[Kre88]{kre:j:optimization}
M.~Krentel.
\newblock The complexity of optimization problems.
\newblock {\em Journal of Computer and System Sciences}, 36(3):490--509, 1988.

\bibitem[Lad75]{lad:j:np-incomplete}
R.~Ladner.
\newblock On the structure of polynomial time reducibility.
\newblock {\em Journal of the ACM}, 22(1):155--171, 1975.

\bibitem[Reg92]{reg:j:diag-unif-fixed-point}
K.~Regan.
\newblock Diagonalization, uniformity, and fixed-point theorems.
\newblock {\em Information and Computation}, 98:1--40, 1992.

\bibitem[RV97]{reg-vol:j:gap}
K.~Regan and H.~Vollmer.
\newblock Gap-languages and log-time complexity classes.
\newblock {\em Theoretical Computer Science}, 188(1--2):101--116, 1997.

\bibitem[Sch82]{sch:j:uni}
U.~Sch{\"{o}}ning.
\newblock A uniform approach to obtain diagonal sets in complexity classes.
\newblock {\em Theoretical Computer Science}, 18:95--103, 1982.

\bibitem[Val79]{val:j:permanent}
L.~Valiant.
\newblock The complexity of computing the permanent.
\newblock {\em Theoretical Computer Science}, 8(2):189--201, 1979.

\bibitem[Zan91]{zan:j:sharp-p}
V.~Zank{\'{o}}.
\newblock \#{P}-completeness via many-one reductions.
\newblock {\em International Journal of Foundations of Computer Science},
  2(1):76--82, 1991.

\end{thebibliography}

%

%

\end{document}